\documentclass[journal]{IEEEtran}
\usepackage[T1]{fontenc}
\usepackage[utf8]{inputenc}
\usepackage{algorithm} 
\usepackage{amsmath,amsfonts}
\usepackage{amssymb, amsthm }
\usepackage{algorithmicx} 
\usepackage{algpseudocode} 
\usepackage{array}
\usepackage{cancel}
\usepackage{subcaption} 
\usepackage{textcomp}
\usepackage{stfloats}
\usepackage{url}
\usepackage{verbatim}
\usepackage{graphicx}
\usepackage{cite}
\usepackage{braket}
\usepackage{mathrsfs}
\usepackage{multirow}
\usepackage{hhline}
\usepackage{makecell}
\usepackage[normalem]{ulem}
\usepackage{setspace}
\usepackage{color}
\usepackage{xcolor}
\usepackage{tikz}
\usetikzlibrary{shapes.geometric, arrows, arrows.meta, positioning, fit, calc}

\pgfdeclarelayer{background}
\pgfsetlayers{background,main}

\definecolor{myred}{RGB}{180,0,0}

\newtheorem{theorem}{Theorem}
\newtheorem{lemma}{Lemma}
\newtheorem{proposition}{Proposition}

\algrenewcommand\algorithmicrequire{\textbf{Input:}}
\algrenewcommand\algorithmicensure {\textbf{Output:}}

\begin{document}

\title{Binary Signal Recovery in Undersampling: Iterative SDP with Majority Voting and Successive Interference Cancellation}

\author{Ece Abay, Burhan Gulbahar, and Fatih Alag\"oz
\thanks{
Burhan Gulbahar is with the Department of Electrical and Electronics Engineering, 
Ya{\c{s}}ar University, 35100 Izmir, T{\"u}rkiye 
(e-mail: burhan.gulbahar@yasar.edu.tr).
Ece Abay was with the Department of Electrical and Electronics Engineering, 
Ya{\c{s}}ar University, 35100 Izmir, T{\"u}rkiye, during the course of this work
(e-mail: ecekose97@gmail.com).
Fatih Alag\"oz is with the Department of Computer Engineering, Bo{\u{g}}azi{\c{c}}i University, 34342 {\.I}stanbul, T{\"u}rkiye
(e-mail: fatih.alagoz@bogazici.edu.tr).
}
\thanks{This work has been submitted to the IEEE for possible publication. Copyright may be transferred without notice, after which this version may no longer be accessible.}
}

\maketitle

\begin{abstract}

Binary compressive sensing (BCS) seeks to recover a $k$-sparse binary vector of length $n$ from $m$ linear measurements. Classical CS guarantees break down  for $m < k$ and convex/greedy BCS algorithms with random Gaussian sensing matrices perform poorly.
We introduce ISDP-MVSIC, which combines randomized semidefinite programming (SDP) sampling, majority voting (MV) and successive interference cancellation (SIC) across $L \ll n$ stages, wrapped in a residual-cost driven retry loop.
The method exposes a tunable complexity--performance trade-off: for $n=100, 144$, raising the worst-case complexity $\mathcal{C}_{max}$ from $7.9 \times 10^9$ to $2.0 \times 10^{10}$ enables empirical exact recovery over $m/k \in [0.4,5.0]$ as the sparsity ratio $s=k/n$ decreases from $0.5$ to $0.1$, by practically targeting the undersampled regime.
\end{abstract}

\section{Introduction}

\IEEEPARstart{B}{inary} compressive sensing (BCS) recovers a sparse $n$-bit binary signal $\mathbf{x}$ with $k$ nonzero elements from measurements $\mathbf{y} = \mathbf{H} \, \mathbf{x}$ where $\mathbf{H} \in \mathbb{R}^{m \times n}$ \cite{donoho2006compressed, candes2006robust, stojnic2010recovery, nakarmi2012bcs}.
For i.i.d. Gaussian $\mathbf{H}$, unique recovery of any $k$-sparse vector requires $m \geq 2k$ measurements \cite{eldar2015compressed}, and more generally $m  = \Omega \big(\delta^{-2}\,k\log(n/k) \big) $ where $\delta\in(0,1)$ is the restricted-isometry constant.
However, many applications require signal recovery in the undersampled and challenging regime $m \leq k$, where CS guarantees break down, including overloaded MIMO \cite{hayakawa2017convex, hajji2020iterative} and underwater sensor networks \cite{wang2019underwater, wu2018compressive}.
Conventional BCS methods offer limited performance for $m \leq k$ \cite{stojnic2010recovery, mangasarian2011probability, wang2013binary, sparrer2016mmse, Fosson2018, fosson2019recovery, hayakawa2020asymptotic, doi2024phase}.

In this article, we propose a flexible algorithm, ISDP-MVSIC targeting the regime $m \leq k$ under i.i.d. Gaussian $\mathbf{H}$ via a tunable complexity--performance trade-off, combining iterative semidefinite programming (SDP), majority voting (MV), and successive interference cancellation (SIC) \cite{wolniansky1998v}.
It randomizes measurements and iteratively solves SDP relaxations to generate candidates via multivariate normal sampling; candidates are filtered by adaptive sparsity and residual cost $C(\widehat{\mathbf{x}})=\|\mathbf{y}-\mathbf{H}\,\widehat{\mathbf{x}}\|_{2}^{2}$, MV identifies reliable bits for decoding, SIC removes their interference, and a residual-driven retry loop wraps the pipeline.
The approach is motivated by \cite{gulbahar2025majority}, which combines cost-minimization sampling, MV, and SIC for MIMO detection. We adapt this machinery to BCS through binary-to-bipolar reduction \eqref{equivalentLS}, a sparsity-aware sample filter, a residual-driven retry loop  and finite-size analysis. 

\subsection{Related Works}

Standard $\ell_1$-norm and mixed-norm methods \cite{stojnic2010recovery, mangasarian2011probability, wang2013binary, hayakawa2020asymptotic, doi2024phase}, greedy algorithms like MMSE-OMP \cite{sparrer2016mmse} and PROMP \cite{flinth2018asparse}, and finite-alphabet approaches \cite{bioglio2014sparse, aissaelbey2015sparsity, hayakawa2017convex} typically degrade for $m \le k$.
Box-SOAV \cite{hayakawa2020asymptotic} extends convex relaxation with asymptotic bounds.
RWR \cite{Fosson2018} and POP \cite{fosson2019recovery} recover only for $m \gtrsim 1.7k$ and $0.9k$.
Annealing-based heuristics on a quadratic unconstrained binary optimization (QUBO) reformulation of \eqref{MLproblem}, including simulated and quantum annealing (QA) \cite{ayanzadeh2019quantum, doi2024phase}, are sensitive to penalty-parameter choice and solver hardware.
Complementary research designs the sensing matrix \cite{nakarmi2012bcs,shirvanimoghaddam2015binary,sarangi2022measurement, romanov2021unsourced}.
These methods differ in their priors: convex and finite-alphabet approaches \cite{wang2013binary, hayakawa2020asymptotic, hayakawa2017convex} exploit no explicit sparsity and target $m\gtrsim k$, whereas sparsity-aware methods \cite{sparrer2016mmse, Fosson2018, fosson2019recovery} require $m\gtrsim0.9$--$1.7k$; ISDP-MVSIC exploits sparsity through its $(k-k_{dec})$-budgeted filter while sustaining $m<k$ (down to $m/k=0.4$). The prior work does not address a stagewise cost/margin model.
Table~\ref{tab:related} positions ISDP-MVSIC against representative estimators. 

\begin{table}[t!]
\caption{Representative BCS estimators on i.i.d. Gaussian $\mathbf{H}$ in the underdetermined regime $m<n$.}
\label{tab:related}
\setlength\tabcolsep{2.5pt}
\renewcommand{\arraystretch}{1.15}
\scriptsize
\centering
\begin{tabular}{|>{\raggedright\arraybackslash}p{2.0cm}|>{\raggedright\arraybackslash}p{1.55cm}|c|>{\raggedright\arraybackslash}p{1.6cm}|c|}
\hline
\textbf{Method (Ref.)} & \textbf{Approach} & \textbf{Year} & \textbf{Tested regime} & \textbf{Complexity} \\
\hline
Weak-$\ell_1$ threshold \cite{stojnic2010recovery} & LP, theoretical PT & 2010 & $m/n \gtrsim 0.5$, $s \approx 0.1$ & $\mathcal{O}(n^{3.5})$ \\
\hline
$\ell_1 + \ell_\infty$ \cite{wang2013binary} & Mixed-norm convex & 2013 & $m > k$ & $\mathcal{O}(n^{3.5})$ \\
\hline
MMSE-OMP \cite{sparrer2016mmse} & Greedy + MMSE & 2016 & $m > k$ & $\mathcal{O}(m n k)$ \\
\hline
RWR \cite{Fosson2018} & Iter. reweighted & 2018 & $m \gtrsim 1.7k$ & $\mathcal{O}(T\,n^{3.5})$ \\
\hline
POP \cite{fosson2019recovery} & Lasserre SDP & 2019 & $m \gtrsim 0.9k$ & $\mathcal{O}(n^{3.5})$ \\
\hline
Box-SOAV \cite{hayakawa2020asymptotic} & Sum-$\ell_1$ + box, DR & 2020 & $m < n$ & $\mathcal{O}(T\,n^{2})$ \\
\hline
QUBO + QA \cite{ayanzadeh2019quantum} & Ising search & 2019 & formulation only & solver dependent \\
\hline
\textbf{ISDP-MVSIC} & Iter. SDP + MV + SIC & 2025 & $m/k \in [0.4,5.0]$, $n \le 144$ & $\mathcal{O}(n^{3.5}\sum_i r_i)$\\
\hline
\end{tabular}
\end{table}

\subsection{Contributions}
 
The novel contributions are listed as follows:
\begin{itemize}\setlength{\itemsep}{0pt}\setlength{\parskip}{0pt}\setlength{\topsep}{1pt}
\item A BCS recovery algorithm combining iterative SDP sampling, MV, SIC, and residual-driven retry, targeting $m \leq k$ under i.i.d. Gaussian $\mathbf{H}$.
\item A tunable complexity--performance trade-off via sample iteration count $r_i$ per stage.
\item Empirical comparison with Box-SOAV \cite{hayakawa2020asymptotic}, RWR \cite{Fosson2018}, MMSE-OMP \cite{sparrer2016mmse}, and theoretical thresholds \cite{stojnic2010recovery, fosson2019recovery}. For $n = 100, 144$, worst-case complexity $\mathcal{C}_{max} \in [7.9\times10^9, 2.0\times10^{10}]$ enables exact recovery for $m/k \in [0.4, 5.0]$ as $s$ varies from $0.5$ to $0.1$, covering regimes inaccessible to listed baselines. 
\end{itemize}

\section{Binary Compressive Sensing}
\label{sec2}

Let $\mathbf{x}_{org} \in \{0,1\}^{n}$ be the unknown signal with $k$ nonzero entries, $k\ll n$.
The measurement model is $\mathbf{y}_{org} = \mathbf{H} \, \mathbf{x}_{org}$, where $\mathbf{H}\in\mathbb{R}^{m\times n}$ and $\mathbf{y}_{org}\in\mathbb{R}^m$.
We focus on $m \le k$.
Applying the affine transformation $\mathbf{x} = \mathbf{1}-2 \, \mathbf{x}_{org}$ maps binary to ${\pm1}$ variables; although $\mathbf{x}$ is dense, it uniquely determines $\mathbf{x}_{org}$ \cite{mangasarian2011probability}.
Substituting into the  model gives the following:
\begin{equation}
\label{equivalentLS}
\mathbf{y} = \mathbf{H} \, \mathbf{x}  \triangleq \mathbf{H} \, \mathbf{1} - 2 \, \mathbf{y}_{org}
\end{equation}

$\mathbf{x}\in\{-1,+1\}^n$ is estimated from $\mathbf{y}$ by recovering $\mathbf{x}_{org}$ from $\mathbf{y}_{org}$  via the following and applying $\hat{\mathbf{x}}_{org}=(\mathbf{1}-\hat{\mathbf{x}})/2$:
\begin{equation}
\label{MLproblem}
\underset{\widehat{\mathbf{x}}\in\{+1,-1\}^{n}}{\text{minimize}} \quad \bigl\lVert\mathbf{y}-\mathbf{H}\,\widehat{\mathbf{x}}\bigr\rVert_{2}^{2} 
\end{equation}

For noisy measurements $\mathbf{y}_{org} = \mathbf{H}\, \mathbf{x}_{org} + \mathbf{w}$, we use $\mathbf{w}\in \mathbb{R}^m$ with components i.i.d. $\mathcal{U}[-\eta,\eta]$, $\eta = 5\times10^{-2}$. Following \cite{fosson2019recovery}, we define $\mathrm{SNR} \approx k/(m\eta^2)$ (19--30 dB for tested sparsity levels) for direct comparability. For context, optimal decoding of $\mathbf{y}_{org}=\mathbf{H}\, \mathbf{x}_{org}$ requires $m \gtrsim \log_2 \binom{n}{k}$ for stable $k$-sparse recovery; for $n=100$, $k=10$ this bound is $\approx 44$, and robust recovery at $m<44$ relies on Gaussian $\mathbf{H}$ structure.

\begin{figure}[t!]
\centering
\includegraphics[width=1\columnwidth]{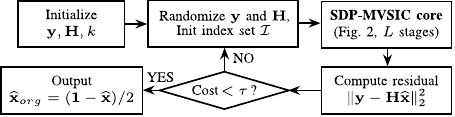}
\caption{The proposed ISDP-MVSIC algorithm. The NO branch triggers a stage-level retry with a fresh randomization, as in the per-stage attempt loop of Algorithm~1.}
\label{fig1}
\end{figure}

\begin{figure}[t!]
\centering
\includegraphics[width=1\columnwidth]{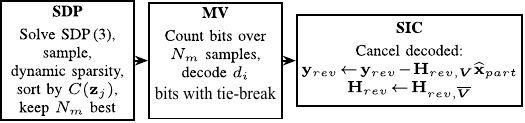}
\caption{SDP\,$\to$\,MV\,$\to$\,SIC pipeline executed at every stage~$i$.}
\label{fig2}
\end{figure}

\begin{figure}[t!]
\centering
\includegraphics[width=1\columnwidth]{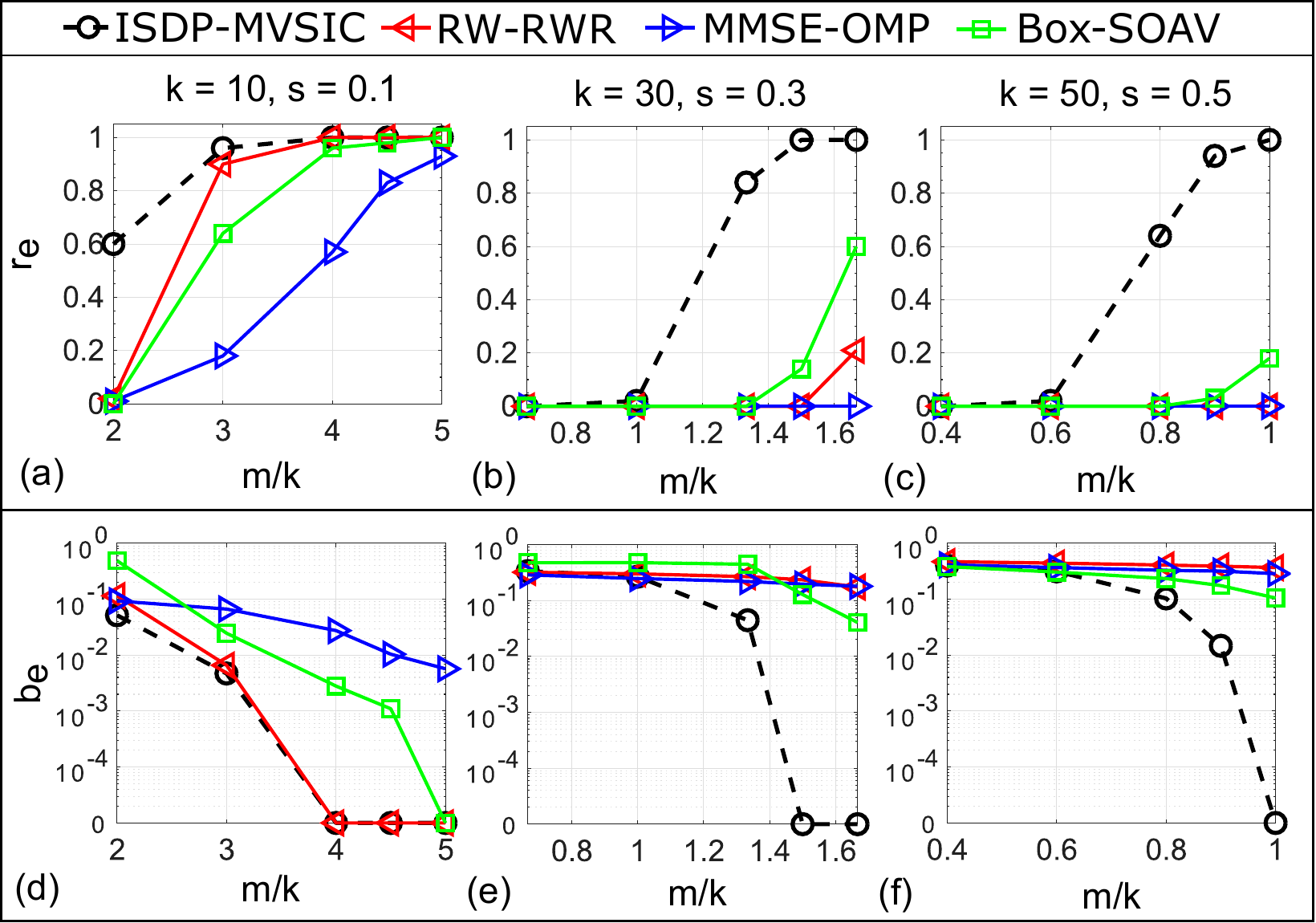}
\caption{Performance comparison (noisy, $n = 100$): $r_e$ in (a)--(c), $b_e$ in (d)--(f) vs. $m/k$ for $k = 10, 30, 50$.}
\label{fig3}
\end{figure}

\begin{figure}[t!]
\centering
\includegraphics[width=1\columnwidth]{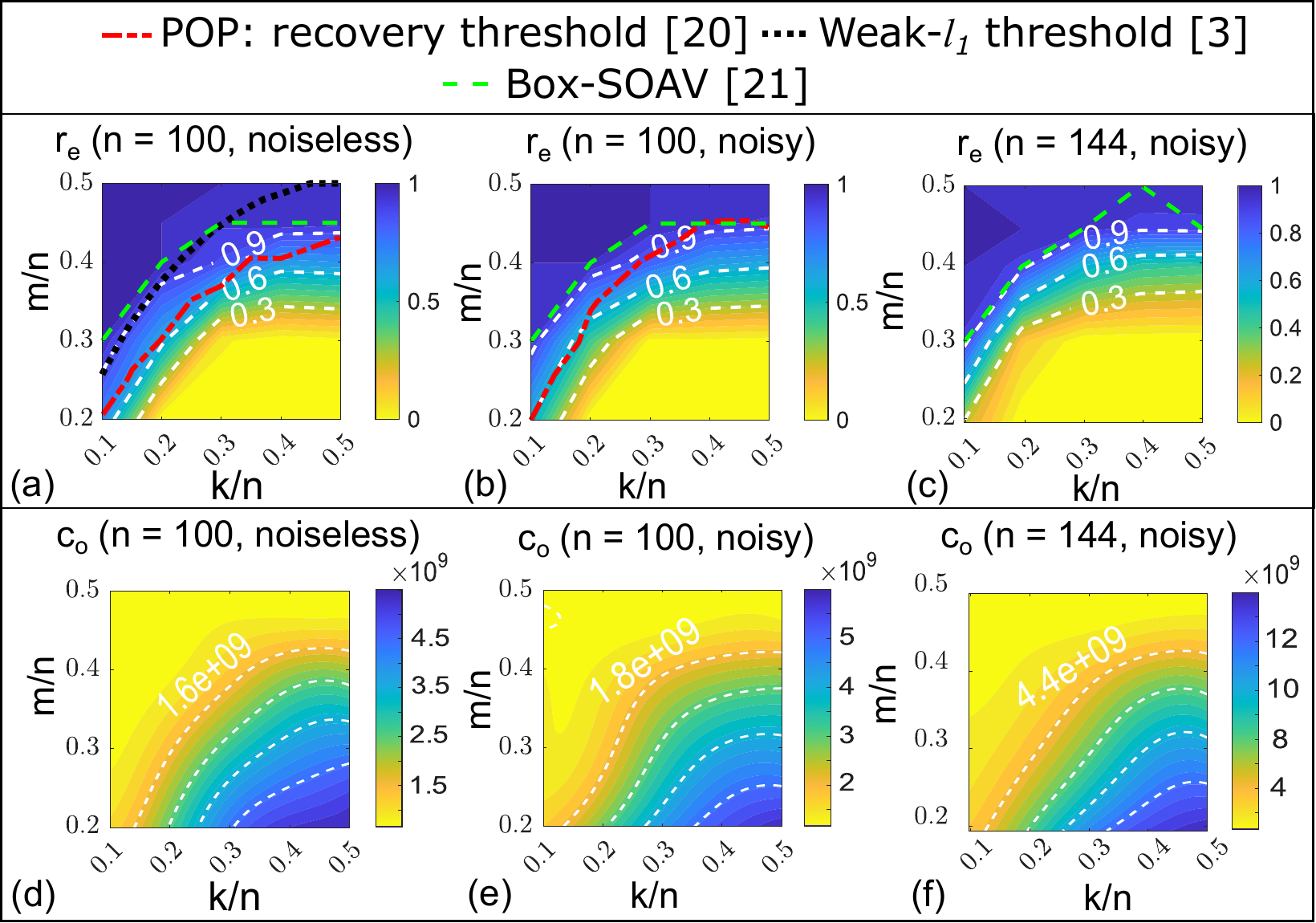}    
\caption{ISDP-MVSIC noiseless simulation results:  $r_e$ and $c_o$ maps for $n = 100$ in (a) and (d). ISDP-MVSIC noisy simulation results: $r_e$ and $c_o$ maps for  $n = 100$ in (b) and (e), and for $n = 144$ in (c) and (f). Lines: POP \cite{fosson2019recovery}, Weak-$\ell_1$ threshold \cite{stojnic2010recovery}, Box-SOAV \cite{hayakawa2020asymptotic} thresholds.}
\label{fig4}
\end{figure} 

\section{ISDP-MVSIC Algorithm}
\label{ISDP-MVSIC}

The ISDP-MVSIC algorithm (Figs.~\ref{fig1}--\ref{fig2}, Algorithm~\ref{alg:isdpmvsic}) uses randomized SDP sampling across $L$ stages with decreasing length $t$ per stage, where $d_{SIC} = [d_1 \hdots d_L]$ and $n = \sum_{i=1}^L d_i$. Unlike the QAOA-based MIMO detector of \cite{gulbahar2025majority}, the sampling here is entirely classical SDP-based and is specialized to sparse BCS by the residual-budgeted filter and the retry loop introduced below.
Before SDP, we left-multiply $\mathbf{y}_{rev}$, $\mathbf{H}_{rev}$ by a fresh random Gaussian matrix $\mathbf{G}\in\mathbb{R}^{t_i\times m}$ at each attempt; the index set $\mathcal{I}$ tracks decoded coordinates so that bits are written back to original positions. The randomization generates diverse candidate pools across the $r_i$ retries; a random tie-break at decoding (Sec.~\ref{ISDP-MVSIC}-B) further decorrelates the bit selection when MV votes saturate. Each stage includes: (1) solve SDP, (2) sample $N_s$ candidates and keep $N_m$ lowest-cost, (3) apply MV to decode $d_i$ reliable bits, (4) SIC removes interference; if the residual cost exceeds $\tau$, the pipeline restarts with a fresh randomization.

\begin{algorithm}[t!]
\caption{ISDP-MVSIC}
\label{alg:isdpmvsic}
\setstretch{0.9}
\begin{algorithmic}[1]
\Require $\mathbf{H}\in\mathbb{R}^{m\times n}$, $\mathbf{y}_{org}$, $k$, $L$, $\{d_i, r_i\}_{i=1}^{L}$, $N_s$, $\Delta$, $\tau$
\Ensure $\widehat{\mathbf{x}}_{org}\in\{0,1\}^{n}$ \textbf{or} \texttt{Failure}
\State $\mathbf{y}\gets \mathbf{H}\mathbf{1}-2\mathbf{y}_{org}$, initialize $\widehat{\mathbf{x}} \gets \mathbf{0}$, $k_{dec} \gets 0$, $\mathbf{y}_{rev}\gets\mathbf{y}$, $\mathbf{H}_{rev}\gets\mathbf{H}$, $\mathcal{I}\gets 1:n$
\For{$i=1:L$} 
  \State $stage\_success \gets \text{false}$
  \For{$attempt = 1:r_i$} 
    \State Draw random Gaussian matrix $\mathbf{G}\in\mathbb{R}^{t_i\times m}$
    \State $\mathbf{H}_{rand} \gets \mathbf{G}\mathbf{H}_{rev}$, $\mathbf{y}_{rand} \gets \mathbf{G}\mathbf{y}_{rev}$
    \State Solve SDP \eqref{SDPrelax} on $(\mathbf{H}_{rand},\mathbf{y}_{rand})$ to obtain $\widehat{\mathbf{S}}$
    \State Draw $N_s$ samples $\tilde{\mathbf{z}}_j\sim\mathcal{N}(\mathbf{0},\widehat{\mathbf{S}})$, take signs $\mathbf{z}_j$
    \State Discard $\mathbf{z}_j$ if count of $-1$ values $> (k - k_{dec})\Delta$; keep the $N_m=d_i+1$ lowest cost
    \State Decode $d_i$ reliable bits via MV with random tie-break to form $\widehat{\mathbf{x}}_{part}$
    \State $\mathbf{y}_{test} \gets \mathbf{y}_{rev} - \mathbf{H}_{rev}(:, \boldsymbol{V})\widehat{\mathbf{x}}_{part}$
    \If{$\|\mathbf{y}_{test}\|_2^2 < \tau$}
      \State $\widehat{\mathbf{x}}_{\mathcal{I}(\boldsymbol{V})} \gets \widehat{\mathbf{x}}_{part}$,\ \, $k_{dec} \gets k_{dec} + (\text{count of } -1\text{ values})$
      \State SIC: $\mathbf{y}_{rev} \gets \mathbf{y}_{test}$, 
      \State $\mathbf{H}_{rev} \gets \mathbf{H}_{rev}(:, \overline{\boldsymbol{V}})$, $\mathcal{I} \gets \mathcal{I}(\overline{\boldsymbol{V}})$
      \State $stage\_success \gets \text{true}$ \textbf{and break}
    \EndIf
  \EndFor
  \If{\textbf{not} $stage\_success$} \Return \texttt{Failure} \EndIf
\EndFor
\State \Return $\widehat{\mathbf{x}}_{org}\gets (\mathbf{1}-\widehat{\mathbf{x}})/2$
\end{algorithmic}
\end{algorithm}

\subsection{Sampling with Randomized SDP}

SDP relaxation has complexity $\mathcal{O}(t_i^{3.5})$  \cite{luo2010sdp, fukuda2019linear}:
\begin{equation}
\label{SDPrelax}
\min_{\mathbf{S}} \operatorname{tr}(\mathbf{Q  \, S })  
 \text { s.t. } \mathbf{ S } \succeq \mathbf{0}, \, S_{\ell, \ell} = 1 \,\, \forall \ell \in [1, t_i+1]
\end{equation}
where $\mathbf{Q}=[\mathbf{H}_{rand}^{T} \mathbf{H}_{rand}, -\mathbf{H}_{rand}^{T} \mathbf{y}_{rand}; -\mathbf{y}_{rand}^{T} \mathbf{H}_{rand}, 0]$.
Samples from $\mathcal{N}(\mathbf{0}, \mathbf{\widehat{S}})$ are quantized after dropping the anchor: $\hat{x}_{j,\ell} = \operatorname{sign}(z_{j,\ell})$ \cite{gulbahar2025majority}.
Dynamic sparsity filtering discards samples with $|\{z_{j,\ell}=-1\}| > (k-k_{dec})\Delta$ where $\Delta=1.1$.

\subsection{Majority Voting}

MV decodes $d_i$ bits from the $N_m=d_i+1$ lowest-cost samples at stage $i$.
Let $\mathbf{z}_j$ denote the $j$-th sample for $j \in [1, N_s]$.
Costs $C(\mathbf{z}_j)$ are sorted, keeping $N_m$ best.
Count $|c_\ell| = |\sum_{j=1}^{N_m} z_{M_j,\ell}|$ is evaluated, sorted descending, and top $d_i$ indices form $\boldsymbol{V}$; when many bits tie at the maximum vote, the order is shuffled before selection so retries decode different bits.
Decoded bits: $\widehat{x}_{part,j} = \operatorname{sign}(c_{V_j})$ for $j \in [1, d_i]$.
Index set $\mathcal{I}(\boldsymbol{V})$ maps to original coordinates; remainder $\mathcal{I}(\overline{\boldsymbol{V}})$ for future stages.

\subsection{Successive Interference Cancellation}
SIC removes decoded columns 
after $\widehat{\mathbf{x}}_{part}$ at indices $\boldsymbol{V}$:
\begin{equation}
\label{eqrev}
\mathbf{y}_{rev} \gets \mathbf{y}_{rev} - \mathbf{H}_{rev, \boldsymbol{V}}  \, \mathbf{\widehat{x}}_{part}; \,\,
 \mathbf{H}_{rev} \gets \mathbf{H}_{rev, \overline{\boldsymbol{V}}} ; \,\, 
\mathcal{I}  \gets \mathcal{I}(\overline{\boldsymbol{V}})
\end{equation}

\section{Computational Complexity}
\label{sec4}

\begin{table}[!t]
\caption{Single-stage complexity of ISDP-MVSIC ($m \times t_i$)}
\label{Table_complex2}
\centering
\scriptsize
\setlength{\tabcolsep}{2.2pt}
\renewcommand{\arraystretch}{0.98}
\resizebox{\columnwidth}{!}{%
\begin{tabular}{|c|c|c|c|}
\hline
\textbf{Task} & \textbf{Complexity} & \textbf{Task} & \textbf{Complexity} \\
\hline
SDP relaxation & $\mathcal{O}(t_i^{3.5})$ & Sorting costs & $\mathcal{O}(N_s\log(N_s))$ \\
\hline
\makecell{Sampling and\\cost calculation} & $\mathcal{O}(N_s t_i^2)$ & MV counting & $\mathcal{O}(N_m t_i)$ \\
\hline
\end{tabular}%
}
\end{table}

The complexity of a single attempt at stage $i$ is dominated by the SDP relaxation $\mathcal{O}(t_i^{3.5})$ and the Gaussian sampling with cost evaluation $\mathcal{O}(N_s t_i^2)$ (Table~\ref{Table_complex2}). Executing the maximum $r_i$ attempts at every stage bounds the total complexity by
\begin{equation}
\mathcal{C}_{max} = \sum_{i=1}^{L} r_i \cdot \Big[ t_i^{3.5} + N_s t_i^2 \Big]
\label{eq:complexity_max}
\end{equation}
Since $t_i \le n$, this is polynomial, $\mathcal{O}(n^{3.5}\sum_i r_i)$, separating the SDP cost from exponential tree search. The empirical cost is lower since a stage exits on its first successful retry, and the pipeline terminates if any stage fails the threshold $\tau$. Let $p_{fail,i}$ be the probability that one attempt at stage $i$ fails $\tau$. A shortened geometric law conditioned on previous-stage success gives the expected number of attempts and complexity as follows:
\begin{equation}
\mathbb{E}[\mathcal{C}_{total}] = \sum_{i=1}^{L} \mathbb{E}[A_i]  \,  \Big[ t_i^{3.5} + N_s t_i^2 \Big]. \qquad \qquad \qquad
\label{eq:complexity_expected}
\end{equation}
\begin{equation}
\mathbb{E}[A_i] = \left( \prod_{j=1}^{i-1} \left[1 - (p_{fail, j})^{r_j}\right] \right)   \frac{1 - (p_{fail, i})^{r_i}}{1 - p_{fail, i}}
\label{eq:expected_attempts}
\end{equation}

In recoverable regimes, early-stage $p_{fail,i}$ stays small, limiting inflation.

\section{Finite-Size Theoretical Model}
\label{sec:theory}

We give a compact conditional model for ISDP-MVSIC: not an unconditional SDP guarantee for $m<k$, but a quantification of how residual-cost ranking, MV, SIC and the design vectors $\mathbf d,\mathbf r$ shape recovery given the per-stage margins measured in simulation. At stage $i$, define
\begin{equation}
D_i=\sum_{j<i}d_j,\quad t_i=n-D_i,\quad \alpha_i=m/t_i,
\label{eq:stage_state_theory}
\end{equation}
where $t_i$ is the number of not-yet-decoded coordinates and $\alpha_i$ is the effective measurement ratio. The stage observation is
\begin{equation}
\mathbf y_i=\mathbf H_i\mathbf x_i+\boldsymbol\varepsilon_i,\qquad
m^{-1}\mathbb E\|\boldsymbol\varepsilon_i\|_2^2=\nu_i ,
\label{eq:stage_obs_theory}
\end{equation}
where $\boldsymbol\varepsilon_i$ collects measurement noise and any interference leaked from earlier SIC errors. Thus $t_i$ decreases and $\alpha_i$ increases across SIC stages, provided previous decisions are correct, so later stages are effectively better conditioned.

\begin{lemma}[Cost separation]
\label{lem:cost_sep_short}
Consider the first stage ($t_1=n$, before any cancellation), with $H_{ab}\sim\mathcal N(0,1)$ and noise of per-coordinate energy $\nu_i$ independent of $\mathbf H_i$. For a candidate $\mathbf z\in\{-1,+1\}^{t_i}$ at Hamming distance $h=d_H(\mathbf z,\mathbf x_i)$,
\begin{equation}
C_i(\mathbf z)=\|\mathbf y_i-\mathbf H_i\mathbf z\|_2^2\overset{d}{=}(4h+\nu_i)\,\chi_m^2 .
\label{eq:cost_chi_short}
\end{equation}
Consequently, for $0<\epsilon<1$,
\begin{equation}
\Pr\!\left\{\left|\frac{C_i(\mathbf z)}{m(4h+\nu_i)}-1\right|\ge\epsilon\right\}
\le 2e^{-m\epsilon^2/8} .
\label{eq:chisq_conc_short}
\end{equation}
\end{lemma}
\begin{proof}
Each row $[\mathbf H_i(\mathbf x_i-\mathbf z)]_a$ sums $h$ nonzero terms $\pm2H_{a\ell}$, hence is Gaussian with variance $4h$; adding the independent perturbation gives variance $4h+\nu_i$, and summing $m$ squared independent Gaussians yields \eqref{eq:cost_chi_short}. The bound \eqref{eq:chisq_conc_short} is the standard chi-square concentration inequality.
\end{proof}

The decoupling is exact at the first stage and approximate later, where SIC and column removal introduce mild dependence; Fig.~\ref{fig5}(a) tests \eqref{eq:cost_chi_short} on a real Gaussian $\mathbf H$ with close agreement. Lemma~\ref{lem:cost_sep_short} justifies cost-sorting: lower-error samples have smaller mean cost, and cost reversals decay exponentially in $m$, so larger $\alpha_i$ improves the accepted pool.
 
Let $\{\mathbf Z_{i,j}^{(r_i)}\}_{j=1}^{N_m}$ denote the accepted low-cost samples at stage $i$ for MV. For the decoded index $\ell\in\boldsymbol V_i$,  sample correctness probability and the worst selected-bit margin are
\begin{equation}
q_{i\ell}^{(r_i)}
=\Pr\{Z_{i,j,\ell}^{(r_i)}=x_{i,\ell}\mid \mathrm{accept}\},\, 
\gamma_i=\min_{\ell\in\boldsymbol V_i}\bigl(2 \,q_{i\ell}^{(r_i)}-1\bigr).
\label{eq:margin_short}
\end{equation}
where we do not model $\gamma_i$ in closed form; it is measured from the accepted-sample vote statistics at each stage. The margin is determined by the weakest decoded bit, since a single incorrect sign fails the stage. If the $N_m r_i$ accepted votes were independent, the effective vote count would be $N_m r_i$; positive correlation between samples taken from the same SDP solution reduces it to an effective count $N_{{\rm eff},i}^{(r_i)}\le N_m r_i$.
\begin{proposition}[MV stage error]
\label{prop:mv_short}
If $\gamma_i>0$ and accepted votes fit Hoeffding-type concentration with effective sample size $N_{{\rm eff},i}^{(r_i)}$, then the probability of missolving at least one $d_i$ bit of the $i$-th satisfies
\begin{equation}
P_i^{(r_i)}\le d_i\exp\!\left(-\tfrac{1}{2}N_{{\rm eff},i}^{(r_i)}\gamma_i^2\right).
\label{eq:stage_fail_short}
\end{equation}
\end{proposition}
The proof is a Hoeffding bound per selected bit and a union bound over the $d_i$ decisions. Qualitatively, stage reliability improves with the margin $\gamma_i$ and the effective vote count, and degrades as more bits $d_i$ are decoded at once, because the weakest selected margin then tends to drop.

\begin{theorem}[]
\label{thm:pipeline_short}
Conditioned on correct previous SIC decisions and positive margins at all stages, the following holds while the second inequality follows from Weierstrass product bound:
\begin{equation}
P_{\rm ex}\ge\prod_{i=1}^{L}(1-P_i^{(r_i)})\ge1-\sum_{i=1}^{L}P_i^{(r_i)} ,
\label{eq:pipe_bound_short}
\end{equation}
\end{theorem}
Theorem~\ref{thm:pipeline_short} assumes error-free SIC. When a stage errs, it perturbs the next residual, which couples the per-stage errors: a wrong sign injects an extra $\pm2\mathbf h_\ell$ term per misdecoded bit, and since stage $i$ decodes $d_i$ bits with error probability $P_i^{(r_i)}$, the expected number of wrong bits is $d_iP_i^{(r_i)}$ and the leaked energy accumulates as
\begin{equation}
\mathbb E[\nu_{i+1}]\le \mathbb E[\nu_i]+4d_iP_i^{(r_i)} .
\label{eq:sic_nu_short}
\end{equation}
Equations \eqref{eq:stage_fail_short}--\eqref{eq:sic_nu_short} make explicit the trade-off that governs the schedule, which we summarize as
\begin{equation}
\min_{\mathbf d,\mathbf r}\ \mathcal C_{\rm pipe}(\mathbf d,\mathbf r)
\quad\mathrm{s.t.}\quad \sum_iP_i^{(r_i)}\le\delta,\quad \sum_i d_i=n,
\label{eq:design_short}
\end{equation}
with $\mathcal{C}_{\rm pipe}$ from \eqref{eq:complexity_expected}. We do not solve \eqref{eq:design_short} in closed form; it formalizes why aggressive $d_i$ lowers later dimensions but raises stage risk, while larger $r_i$ improves reliability at near-linear cost. The schedules $\mathbf{d}_{\mathrm{SIC}},\mathbf{r}_{\mathrm{SDP}}$ are set empirically, but the margin collapse measured in Fig.~\ref{fig5}(b) supports the chosen profile: more repetition in the early stages, where the margins are smallest, and fewer attempts once they saturate. The model therefore guides schedule design instead of only describing it.

\begin{figure}[!t]
\centering
\includegraphics[width=1\columnwidth]{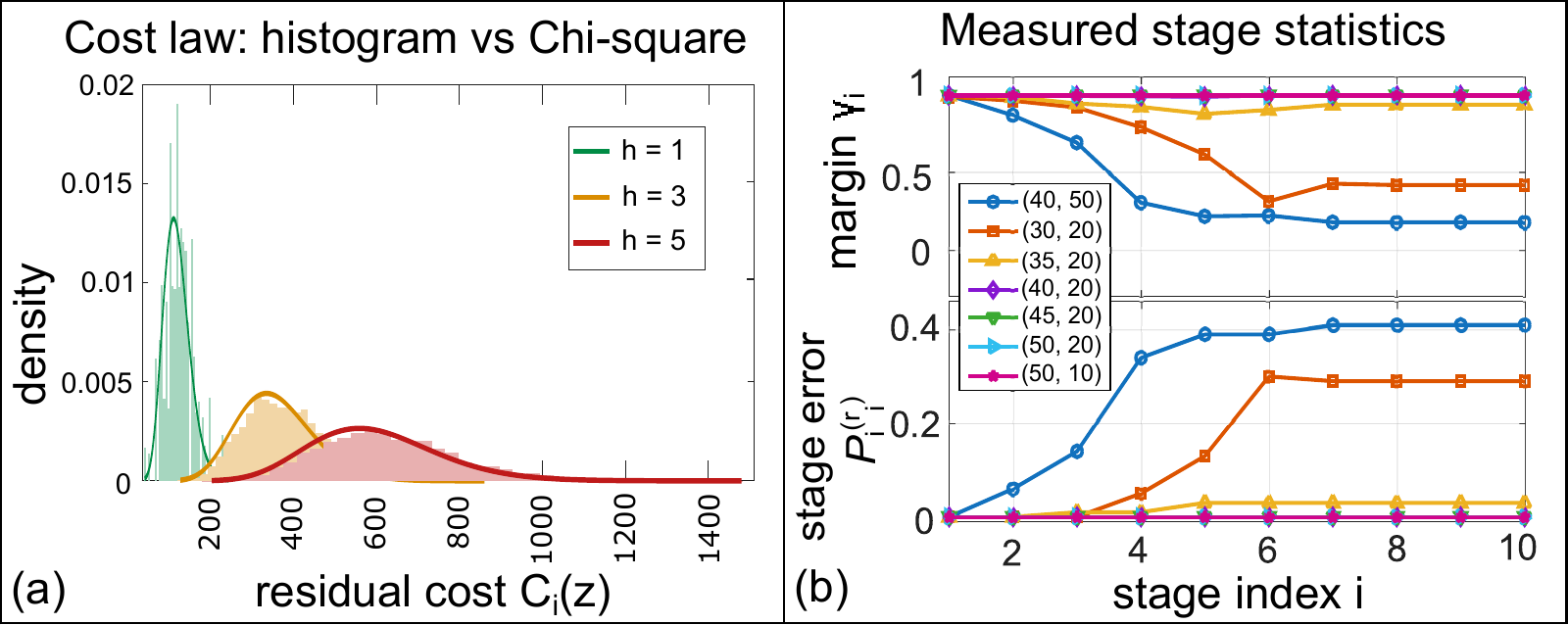}
\caption{Empirical validation on a real Gaussian $\mathbf H$: (a) cost histograms for fixed $h$ match the predicted $(4h+\nu_i)\chi_m^2$ law of Lemma~\ref{lem:cost_sep_short}, with mean cost linear in $h$; (b) measured margin $\gamma_i$ and stage error $P_i^{(r_i)}$ versus stage for $(m,k)\in\{(50,10),(50,20),(45,20),(40,20),(35,20),(30,20),\allowbreak(40,50)\}$, $n=100$.}
\label{fig5}
\end{figure}

Fig.~\ref{fig5} validates the model: (a) mean cost is linear in $h$, and (b) $\gamma_i$ stays near one with near-zero stage error for $m/k\gtrsim1.75$ but collapses over the SIC stages at smaller $m/k$ (sharply for $(30,20)$ and the undersampled $(40,50)$), as predicted by Proposition~\ref{prop:mv_short} and \eqref{eq:sic_nu_short}. The strongly correlated decoded-bit votes ($N_{{\rm eff},i}^{(r_i)}\approx1$) make \eqref{eq:stage_fail_short} diagnostic, not tight.

\section{Numerical Simulations}
\label{sec6}

We use $\mathbf d_{\mathrm{SIC}}=[10^{\times10}]$ and $[10^{\times14},4]$ for $n=100$ and $144$, with $\mathbf r_{\mathrm{SDP}}=[20,10^{\times5},2^{\times4}]$ and $[20,10^{\times5},2^{\times9}]$ ($a^{\times q}$ = $q$ repetitions of $a$), giving worst-case complexity $\mathcal{C}_{max} = 7.9\times10^9, 2.0\times10^{10}$ (Eq.~\eqref{eq:complexity_max}). We compare $r_e$ and BER $b_e$ for ISDP-MVSIC, Box-SOAV \cite{hayakawa2020asymptotic}, RWR \cite{Fosson2018}, and MMSE-OMP \cite{sparrer2016mmse} under noisy measurements; Box-SOAV uses $\mathbf{A} \sim \mathcal{N}(0, 1/n)$ and ISDP-MVSIC uses $\mathbf{H} \sim \mathcal{N}(0, 1)$ (both standard). Parameters: $M=100$ trials, $N_s=16384$, $N_m=d_i+1$ (the smallest pool that accepts a strict majority over the $d_i$ decoded bits), $\tau=10^{-4}$/$11$ (noiseless/noisy), max retries $20$, $s\in[0.1,0.5]$, $m/n\in[0.2,0.5]$, with $r_e=M_t/M$, $b_e=\sum e_i/(Mn)$; contours (95\% CI $\pm 0.10$ at $r_e = 0.5$) are empirical, not asymptotic. A complete implementation reproducing all figures is publicly available as a Code Ocean compute capsule \cite{computecapsule}.

Fig.~\ref{fig3} compares noisy recovery at $n=100$ for $k=10,30,50$: ISDP-MVSIC reaches high $r_e$ (top row) and lower BER (bottom row) at smaller $m/k$ than the baselines, most visibly in the severe regime $m/k\le1.5$ where they fail or accept large BER. Fig.~\ref{fig4} gives the $(k/n,m/n)$ recovery and complexity maps; in the noiseless panel (a) the red dash-dot, black dotted, and green dashed curves are the POP \cite{fosson2019recovery}, Weak-$\ell_1$ threshold \cite{stojnic2010recovery}, and Box-SOAV \cite{hayakawa2020asymptotic} thresholds. The complexity maps report the realized cost $c_o$, which stays well below the worst case $\mathcal{C}_{max}$ (e.g.\ $\sim\!7\times10^8$ vs $7.9\times10^9$ at $n=100$) because successful stages exit their retry loop early. ISDP-MVSIC recovers below all three thresholds, including $m<k$, at the higher polynomial cost of Table~\ref{tab:related}. We emphasize that this is a feasibility comparison rather than a runtime one: for $m<k$ the convex and greedy baselines fail to recover reliably, so the higher cost of ISDP-MVSIC is the trade-off for recovering where they cannot. The noisy recovery maps (b),(c) shrink as expected but stay below the POP and Box-SOAV curves over much of the grid, so the gain persists under bounded noise. Although we report $n=100,144$, behavior is governed by $k/n$ and $m/n$: maps stay similar at matched ratios while complexity scales with $n$ via \eqref{eq:complexity_max}.

We test a single benign noise level ($\eta=5\times10^{-2}$, $\approx$19--30~dB); at substantially lower SNR the cost separation $4h+\nu_i$ of Lemma~\ref{lem:cost_sep_short} degrades, so the recovery region is expected to contract, and a full SNR sweep is left to future work.

\section{Conclusion}
\label{sec7}

ISDP-MVSIC offers a tunable complexity--performance trade-off for BCS at $m < k$, combining SDP sampling, MV, SIC, and retry. Future work includes automated $r_{\text{SDP}}$, $d_{\text{SIC}}$ optimization, budget-matched comparisons with alternative samplers (QAOA, annealing), component ablations, and theoretical recovery guarantees.

\end{document}